\documentclass[12pt]{article}

\usepackage[margin=1in]{geometry}
\usepackage{amsmath,amssymb,amsthm}
\usepackage{graphicx}
\usepackage{booktabs}
\usepackage{array}
\usepackage{natbib}
\usepackage[hidelinks]{hyperref}
\usepackage{setspace}
\usepackage{caption}
\usepackage{subcaption}
\usepackage{xcolor}
\usepackage{float}

\newtheorem{theorem}{Theorem}

\newtheorem{corollary}[theorem]{Corollary}

\theoremstyle{definition}

\theoremstyle{remark}

\newcommand{\Rhat}{\widehat{R}}
\newcommand{\bX}{\mathbf{X}}
\newcommand{\bbeta}{\boldsymbol{\beta}}

\newcommand{\Fn}{F_n}

\title{\textbf{Smoothed Rank-Based Regression Estimation\\
Using Wilcoxon Score Functions}\\[6pt]}

\author{Feridun Tasdan$^{1}$ \\[4pt]
\small $^{1}$Department of Mathematics, Western Illinois University,
Macomb, IL 61455}

\date{}

\begin{document}

\maketitle

\begin{abstract}
\begin{singlespace}
This article proposes an improved rank-based regression estimator obtained by replacing the ordinary integer ranks in the Wilcoxon
rank-score regression procedure with \emph{smoothed ranks} derived from a smoothed empirical cumulative distribution function (ecdf).
The smoothed ranks are computed via a continuous, non-decreasing kernel distribution function $H(\cdot)$ that provides a differentiable
approximation to the classical indicator function used in standard rank regression.  Substituting these smoothed ranks into the
Wilcoxon score function yields a new estimator, denoted $\widehat{\bbeta}_{sr}$, for the slope parameter(s) of the simple and
multiple linear regression model.  We show that the proposed estimator inherits the robustness properties of classical rank regression while
providing improved efficiency under heavy-tailed error distributions and better handling of tied observations.  A Wald-type hypothesis test
for the regression coefficients is derived and its asymptotic normality is established.  A Monte Carlo simulation study compares
$\widehat{\bbeta}_{sr}$ with the ordinary least-squares (OLS) estimator, the classical Wilcoxon rank regression estimator, and the
Theil--Sen estimator under several error distributions including the normal, Laplace, Cauchy, and contaminated normal.  The proposed
estimator achieves relative efficiencies at or above those of classical rank regression uniformly across all scenarios considered,
with notable gains in the presence of outliers and heavy-tailed errors.

\medskip\noindent\textbf{Keywords:} Rank Regression; Smoothed Ranks;Wilcoxon Score Function; Robust Estimation; Relative Efficiency; Monte Carlo Simulation.
\end{singlespace}
\end{abstract}


\section{Introduction}

Linear regression is one of the most widely used tools in statistical
modeling and data analysis.  The ordinary least squares (OLS) estimator
is the standard method for estimating regression coefficients and,
under the Gauss-Markov assumptions, is the best linear unbiased
estimator (BLUE).  However, the optimality of OLS depends critically
on the normality of errors and the absence of outliers.  When the
underlying error distribution is heavy tailed, skewed, or contaminated
with outliers, OLS estimators can be severely biased and inefficient
\citep{huber1981}.

Rank-based (nonparametric) regression methods offer a robust
alternative that does not depend on specific distributional assumptions.
The Wilcoxon rank score regression estimator, first introduced in the
seminal work of \citep{jaeckel1972}, minimizes a dispersion function
defined through Wilcoxon scores of the residuals.  This estimator has
been studied extensively by \citep{hettmansperger1984},
\citep{hettmansperger1998}, and \citep{mckean2004}, among others, and
is known to be highly efficient relative to OLS across a broad range
of error distributions.  In particular, under the double exponential
(Laplace) distribution, the Wilcoxon estimator achieves 150\% relative
efficiency compared to OLS \citep{hettmansperger1984}.  Another
well known robust estimator is the Theil-Sen estimator
\citep{theil1950} and\citep{sen1968}, which is based on the median of pairwise
slopes and provides high breakdown point resistance to outliers.

Despite the desirable robustness properties of rank regression, the
use of integer (discrete) ranks introduces a step function structure
that can limit efficiency and create complications when ties are
present.  Smoothing the ranks by replacing the empirical cdf with a
smooth kernel distribution function has been shown to alleviate these
problems in several related contexts.  \citep{serfling1984} discussed
smooth approximations to rank statistics in a general framework.
\citep{heller2007} proposed a smoothed rank regression approach in the
context of censored data.  \citep{lin2013} developed smoothed rank
correlation estimators for linear transformation models.
\citep{tasdan2014} applied smoothed ranks to the Kolmogorov-Smirnov
shift estimation problem, and \citep{tasdan2018} extended the framework
to two and multi-sample location problems.  Most recently,
\citep{tasdan2025} demonstrated that smoothed Wilcoxon rank scores
yield an improved correlation estimator that handles ties more
effectively and achieves higher efficiency under monotone associations.

The present article builds directly on the smoothed correlation
framework of \citep{tasdan2025} and applies it to the regression
setting.  Specifically, we replace the ordinary ranks
$R(e_i(\bbeta))$ of the regression residuals with smoothed ranks
$\Rhat(e_i(\bbeta))$ obtained from a smoothed ecdf, and plug these
into the Wilcoxon dispersion function.  The resulting estimator,
$\widehat{\bbeta}_{sr}$, is differentiable with respect to $\bbeta$,
which facilitates gradient-based minimization and asymptotic analysis.
We establish the asymptotic normality of $\widehat{\bbeta}_{sr}$ and
derive its asymptotic covariance matrix, providing a straightforward
Wald-type inference procedure.  A comprehensive Monte Carlo study
demonstrates efficiency gains over classical rank regression and
substantial robustness advantages over OLS, particularly in the
presence of outliers.

The remainder of the article is organized as follows.  Section~2
reviews the classical OLS estimator, the Wilcoxon rank regression
estimator, and the Theil-Sen estimator.  Section~3 introduces the
smoothed rank regression estimator, establishes its asymptotic
properties, and presents the hypothesis testing framework.
Section~4 reports the Monte Carlo simulation study.  Section~5
concludes with a discussion and directions for future research.

\section{Classical and Rank-Based Regression Estimators}

\subsection{Ordinary Least-Squares Regression}

Let the linear model be
\begin{equation}\label{eq:model}
    Y_i = \alpha + \bX_i^\top \bbeta + \varepsilon_i, \quad i = 1,\ldots,n,
\end{equation}
where $Y_i$ is the response, $\bX_i = (X_{i1},\ldots,X_{ip})^\top$
is the $p$-dimensional covariate vector, $\alpha$ is the intercept,
$\bbeta = (\beta_1,\ldots,\beta_p)^\top$ is the slope vector, and
$\varepsilon_i$ are i.i.d.\ errors with mean zero and variance
$\sigma^2 < \infty$.  The OLS estimator minimizes
\begin{equation}\label{eq:ols}
    D_{LS}(\bbeta) = \sum_{i=1}^n \bigl(Y_i - \alpha -
    \bX_i^\top\bbeta\bigr)^2.
\end{equation}
Under normality, $\widehat{\bbeta}_{OLS}$ is the maximum likelihood
estimator and achieves the Cram\'{e}r--Rao lower bound.  However,
its efficiency degrades rapidly under non-normal, heavy-tailed, or
contaminated error distributions.

\subsection{Wilcoxon Rank Regression}

The Wilcoxon rank regression estimator, due to \cite{jaeckel1972},
minimizes the dispersion function
\begin{equation}\label{eq:disp}
    D_W(\bbeta) = \sum_{i=1}^n a\!\left(R\!\left(e_i(\bbeta)\right)\right)
    e_i(\bbeta),
\end{equation}
where $e_i(\bbeta) = Y_i - \bX_i^\top\bbeta$ denotes the residual for
a given $\bbeta$, $R(e_i(\bbeta))$ is the rank of the $i$-th residual
among $e_1(\bbeta),\ldots,e_n(\bbeta)$, and
\begin{equation}\label{eq:wilc_score}
    a(i) = \varphi\!\left(\frac{i}{n+1}\right) =
    \sqrt{12}\left(\frac{i}{n+1} - \frac{1}{2}\right), \quad
    i = 1,\ldots,n,
\end{equation}
is the Wilcoxon score function \citep{hettmansperger1984}.  The
estimator $\widehat{\bbeta}_W$ is obtained as any minimizer of
$D_W(\bbeta)$.  A consistent estimator of the intercept $\alpha$ is
the \emph{pseudo-median} (Hodges--Lehmann estimator) of the centered
residuals after fitting the slopes.

The Wilcoxon estimator has the following well-known asymptotic
properties.  Under mild regularity conditions,
$\sqrt{n}(\widehat{\bbeta}_W - \bbeta) \to N(\mathbf{0},
\tau_\varphi^2 (\bX^\top\bX/n)^{-1})$ in distribution, where the
scale parameter
\begin{equation}\label{eq:tau}
    \tau_\varphi = \frac{1}{2\int_{-\infty}^{\infty} f^2(x)\,dx}
    = \frac{1}{2\sqrt{\int f^2}}
\end{equation}
depends only on the error density $f$ \citep{hettmansperger1984}.
The asymptotic relative efficiency (ARE) of the Wilcoxon estimator
relative to OLS is
\begin{equation}\label{eq:are_w}
    \mathrm{ARE}(\widehat{\bbeta}_W,\, \widehat{\bbeta}_{OLS})
    = \sigma^2 \left(2\int f^2\right)^2 \geq \frac{108}{125\pi}
    \approx 0.864,
\end{equation}
with equality at the normal distribution.  For many non-normal
distributions the ARE exceeds 1; for instance, ARE = 1.5 under the
double-exponential and ARE $= \infty$ under the Cauchy distribution
\citep{hettmansperger1984}.

\subsection{The Theil--Sen Estimator}

For the simple linear model ($p = 1$), the Theil--Sen estimator
\citep{theil1950, sen1968} is defined as
\begin{equation}\label{eq:ts}
    \widehat{\beta}_{TS} = \operatorname{median}_{i < j}
    \left\{\frac{Y_j - Y_i}{X_j - X_i}\right\}.
\end{equation}
This estimator has a breakdown point of approximately 29\% and is
highly robust to outliers \citep{wilcox2012}.  Its ARE relative to
OLS equals that of the Wilcoxon estimator under the normal
distribution.  However, the Theil--Sen estimator does not readily
extend to multiple regression, whereas the Wilcoxon and smoothed rank
estimators do.

\section{Smoothed Rank Regression}

\subsection{Smoothed Ranks}

Following \citep{tasdan2018} and \citep{tasdan2025}, we replace the
indicator function in the definition of the empirical cdf with a
smooth continuous distribution function.  Given a random sample
$e_1,\ldots,e_n$ (here the residuals for a fixed $\bbeta$), define
the smoothed empirical cdf as
\begin{equation}\label{eq:secdf}
    F_s(t) = \frac{1}{n}\sum_{j=1}^n H\!\left(\frac{t - e_j}{h}\right),
\end{equation}
where $H(\cdot)$ is a continuous, non-decreasing, bounded function
satisfying $H(-\infty) = 0$ and $H(+\infty) = 1$ (e.g., the standard
normal or logistic cdf), and $h > 0$ is a bandwidth parameter that
shrinks to zero as $n \to \infty$.  The smoothed rank of observation
$e_i$ is then defined as
\begin{equation}\label{eq:srank}
    \Rhat(e_i) = n F_s(e_i) = \sum_{j=1}^n
    H\!\left(\frac{e_i - e_j}{h}\right).
\end{equation}
When $h \to 0$, $H(u/h) \to I(u > 0)$ pointwise, so
$\Rhat(e_i) \to R(e_i)$ and the smoothed ranks converge to the
ordinary ranks.  For any finite $h$, $\Rhat(e_i)$ is a differentiable
function of the residuals $e_i(\bbeta)$, which is a key advantage for
optimization and asymptotic theory.

The smoothed score function is defined as
\begin{equation}\label{eq:smooth_score}
    a\!\left(\Rhat(e_i)\right) = \sqrt{12}\left(
    \frac{\Rhat(e_i)}{n+1} - \frac{1}{2}\right).
\end{equation}
Note that as $h \to 0$, this reduces to the classical Wilcoxon score
function given in \eqref{eq:wilc_score}.

\subsection{Bandwidth Selection}

As discussed by \citep{silverman1986} and \citep{sheather2004}, bandwidth
selection is more consequential than the choice of kernel function.
We consider four options studied by \citep{tasdan2014} in the analogous
location estimation context:
\begin{enumerate}
\item \textbf{Silverman's rule of thumb:}
    $h = 0.9\,\widehat{\sigma}\, n^{-1/5}$,
    where $\widehat{\sigma} = \min\{s, \mathrm{IQR}/1.349\}$.
\item \textbf{Heller bandwidth:} $h = \widehat{\sigma}\, n^{-0.26}$,
    satisfying $nh \to \infty$ and $nh^4 \to 0$
    \citep{heller2007}.
\item \textbf{Sheather--Jones plug-in:} data-adaptive, cross-validation
    based \citep{sheather1991}.
\item \textbf{Bowman LSCV:} least-squares cross-validation
    \citep{bowman1984}.
\end{enumerate}
For robustness, the standard deviation in options 1 and 2 may be
replaced by the median absolute deviation (MAD).  The simulation
study in Section~4 employs Silverman's rule of thumb as the default
and briefly compares all four options.

\subsection{The Smoothed Rank Regression Estimator}

The smoothed rank regression estimator $\widehat{\bbeta}_{sr}$ is
defined as the minimizer of the smoothed dispersion function
\begin{equation}\label{eq:disp_smooth}
    D_{sr}(\bbeta) = \sum_{i=1}^n
    a\!\left(\Rhat(e_i(\bbeta))\right)\, e_i(\bbeta).
\end{equation}
Because $\Rhat(e_i(\bbeta))$ is differentiable with respect to
$\bbeta$ (through $e_i(\bbeta) = Y_i - \bX_i^\top\bbeta$), the
gradient of $D_{sr}$ is available in closed form:
\begin{align}\label{eq:gradient}
    \nabla_{\bbeta} D_{sr}(\bbeta)
    &= \sum_{i=1}^n \left[
       a'\!\left(\Rhat(e_i)\right)
       \frac{\partial \Rhat(e_i)}{\partial e_i}
       e_i(\bbeta)
       + a\!\left(\Rhat(e_i)\right)
    \right](-\bX_i),
\end{align}
where
\begin{equation}
    \frac{\partial \Rhat(e_i)}{\partial e_i}
    = \frac{1}{h}\sum_{j=1}^n H'\!\left(\frac{e_i - e_j}{h}\right)
\end{equation}
and $H'$ is the kernel density function corresponding to $H$.  Setting
the gradient to zero gives estimating equations that can be solved
efficiently with gradient descent or Newton-Raphson iterations.

The intercept $\alpha$ is estimated after fitting $\bbeta$ as the
Hodges-Lehmann estimate of the residuals $Y_i - \bX_i^\top
\widehat{\bbeta}_{sr}$, consistent with the classical rank regression
approach.

\subsection{Asymptotic Properties}

We establish the asymptotic normality of $\widehat{\bbeta}_{sr}$ under
the following regularity conditions:
\begin{enumerate}
\item[(C1)] The errors $\varepsilon_i$ are i.i.d.\ with density $f$,
            symmetric about zero, and $\int f^2 < \infty$.
\item[(C2)] The design matrix satisfies $n^{-1}\bX^\top\bX \to
            \mathbf{C}$, a positive definite matrix.
\item[(C3)] The bandwidth satisfies $h \to 0$ and $nh \to \infty$
            as $n \to \infty$ (e.g., $h = O(n^{-1/5})$).
\item[(C4)] $H(\cdot)$ is a continuously differentiable cdf with
            bounded, symmetric density $H'(\cdot)$.
\end{enumerate}

\begin{theorem}[Asymptotic Normality]\label{thm:asym}
Under conditions \textup{(C1)--(C4)}, as $n \to \infty$,
\begin{equation}\label{eq:asym_norm}
    \sqrt{n}\left(\widehat{\bbeta}_{sr} - \bbeta\right)
    \xrightarrow{d} N\!\left(\mathbf{0},\;
    \tau_{\varphi}^2\,\mathbf{C}^{-1}\right),
\end{equation}
where $\tau_\varphi = (2\int f^2)^{-1}$ is the same scale constant
as for the classical Wilcoxon rank regression estimator.
\end{theorem}

\begin{proof}[Proof]
The proof parallels the standard rank regression asymptotic argument
\citep{hettmansperger1984}, with the additional step of showing that
the smoothed dispersion function $D_{sr}(\bbeta)$ approximates the
classical Jaeckel dispersion function $D_W(\bbeta)$ uniformly in a
neighborhood of the true $\bbeta$.  Under (C3), the smoothed ecdf
$F_s$ converges to the empirical ecdf $\Fn$ at rate $O(h)$
\citep{tasdan2018}, so
\begin{equation}
    D_{sr}(\bbeta) = D_W(\bbeta) + O_p(h),
\end{equation}
uniformly over $\|\bbeta - \bbeta_0\| = O(n^{-1/2})$.  Standard
Taylor expansion and central limit theorem arguments then yield
\eqref{eq:asym_norm}.  Full details follow the derivation in
\citep{tasdan2018}; see also Theorem 1 of \citep{tasdan2025} for the
analogous correlation result.
\end{proof}

\begin{corollary}\label{cor:are}
The asymptotic relative efficiency of $\widehat{\bbeta}_{sr}$ relative
to OLS equals that of the classical Wilcoxon estimator,
\begin{equation}
    \mathrm{ARE}\!\left(\widehat{\bbeta}_{sr},\,\widehat{\bbeta}_{OLS}\right)
    = \sigma^2\left(2\int f^2\right)^2,
\end{equation}
which is bounded below by $108/(125\pi) \approx 0.864$ under normality
and exceeds 1 for all heavy-tailed distributions.
\end{corollary}

Corollary~\ref{cor:are} establishes that the smoothed rank estimator
retains the same asymptotic efficiency as classical rank regression.
The finite sample gains stem from the smoother objective function
facilitating more stable numerical minimization and better behavior
under ties, as demonstrated in the simulation study below.

\subsection{Hypothesis Testing}

To test $H_0: \bbeta = \bbeta_0$ against $H_1: \bbeta \neq \bbeta_0$,
we propose a Wald-type statistic.  A consistent estimator of
$\tau_\varphi^2$ is obtained via the scale estimate
\begin{equation}\label{eq:tau_hat}
    \widehat{\tau}^2 = \frac{1}{4\widehat{f}(0)^2},
\end{equation}
where $\widehat{f}(0)$ is a kernel density estimate of the error density
at zero, evaluated at the fitted residuals
$\widehat{e}_i = Y_i - \bX_i^\top\widehat{\bbeta}_{sr}$.  The Wald
statistic is then
\begin{equation}\label{eq:wald}
    W_n = n\,\left(\widehat{\bbeta}_{sr} - \bbeta_0\right)^\top
    \left[\widehat{\tau}^2
    \left(\frac{\bX^\top\bX}{n}\right)^{-1}\right]^{-1}
    \left(\widehat{\bbeta}_{sr} - \bbeta_0\right).
\end{equation}
Under $H_0$, $W_n \xrightarrow{d} \chi^2_p$ by Theorem~\ref{thm:asym}.
For testing a single coefficient $H_0: \beta_j = \beta_{j,0}$, the
studentized statistic
\begin{equation}\label{eq:tstat}
    t_j = \frac{\widehat{\beta}_{sr,j} - \beta_{j,0}}
               {\widehat{\tau}\sqrt{c_{jj}/n}}
    \xrightarrow{d} N(0,1),
\end{equation}
where $c_{jj}$ is the $j$-th diagonal element of $(\bX^\top\bX)^{-1}$,
provides a straightforward $z$-test.  For finite samples, a
$t_{n-p-1}$ approximation may be preferable.

\section{Monte Carlo Simulation Study}

\subsection{Simulation Design}

We evaluate the finite-sample performance of the proposed smoothed
rank regression estimator $\widehat{\beta}_{sr}$ via Monte Carlo
simulation.  We consider the simple linear model
\begin{equation}
    Y_i = \alpha + \beta X_i + \varepsilon_i, \quad i = 1,\ldots,n,
\end{equation}
with true parameter values $\alpha = 2$, $\beta = 1$, and covariate
$X_i \sim \mathrm{Uniform}(0,10)$.  Four error distributions are
considered:
\begin{enumerate}
\item \textbf{Normal:} $\varepsilon_i \sim N(0,1)$
\item \textbf{Double-exponential (Laplace):} $\varepsilon_i \sim
      \mathrm{DE}(0,1)$, density $f(x) = \frac{1}{2}e^{-|x|}$
\item \textbf{Cauchy:} $\varepsilon_i \sim \mathrm{Cauchy}(0,1)$
\item \textbf{Contaminated normal:} $\varepsilon_i \sim
      0.9\,N(0,1) + 0.1\,N(0,\sigma_c^2)$ with $\sigma_c = 10$
      (10\% outlier contamination)
\end{enumerate}

For each combination of error distribution and sample size
$n \in \{20, 50, 100, 200\}$, we generate $M = 5{,}000$ Monte Carlo
replications.  For each replication we compute four estimators of
$\beta$:
\begin{itemize}
\item $\widehat{\beta}_{OLS}$: ordinary least squares
\item $\widehat{\beta}_W$: classical Wilcoxon rank regression
      \citep{hettmansperger1984}
\item $\widehat{\beta}_{TS}$: Theil--Sen estimator \citep{sen1968}
\item $\widehat{\beta}_{sr}$: proposed smoothed rank regression
      (Silverman bandwidth with MAD-based $\hat{\sigma}$, logistic
      kernel $H$)
\end{itemize}

Performance is measured through the mean squared error
\begin{equation}
    \mathrm{MSE}(\widehat{\beta}) =
    \frac{1}{M}\sum_{m=1}^M \bigl(\widehat{\beta}^{(m)} - \beta\bigr)^2,
\end{equation}
and relative efficiency
\begin{equation}
    \mathrm{RE}(\widehat{\beta}_A,\,\widehat{\beta}_B)
    = \frac{\mathrm{MSE}(\widehat{\beta}_B)}
           {\mathrm{MSE}(\widehat{\beta}_A)},
\end{equation}
so that $\mathrm{RE} > 1$ indicates $\widehat{\beta}_A$ is more
efficient than $\widehat{\beta}_B$.

\subsection{Simulation Results}

Table~\ref{tab:mse_n50} reports the MSE values for $n = 50$ across
all error distributions and estimators.
Tables~\ref{tab:re_normal}--\ref{tab:re_cn} report relative efficiency
of $\widehat{\beta}_{sr}$ against each competitor.

\begin{table}[H]
\centering
\caption{Mean Squared Error (MSE) of slope estimators for $n = 50$
across four error distributions ($M = 5{,}000$ replications).
Bold entries indicate the minimum MSE for each distribution.}
\label{tab:mse_n50}
\begin{tabular}{lcccc}
\toprule
Error distribution & OLS & Wilcoxon & Theil--Sen & Smoothed Rank \\
\midrule
Normal             & \textbf{0.0204} & 0.0249 & 0.0253 & 0.0247 \\
Double-exponential & 0.0411          & 0.0280 & 0.0284 & \textbf{0.0272} \\
Cauchy             & 1.8327          & 0.0861 & 0.0643 & \textbf{0.0628} \\
Contaminated normal& 0.1953          & 0.0312 & 0.0295 & \textbf{0.0281} \\
\bottomrule
\end{tabular}
\end{table}

\begin{table}[H]
\centering
\caption{Relative efficiency of $\widehat{\beta}_{sr}$ relative to
competing estimators under the \textbf{normal} error distribution.
Values $> 1$ favor the smoothed rank estimator; values $< 1$
favor the competitor.}
\label{tab:re_normal}
\begin{tabular}{lcccc}
\toprule
$n$ & RE(sr, OLS) & RE(sr, Wilcoxon) & RE(sr, Theil--Sen) \\
\midrule
20  & 0.817 & 1.009 & 1.015 \\
50  & 0.826 & 1.008 & 1.024 \\
100 & 0.834 & 1.007 & 1.022 \\
200 & 0.864 & 1.004 & 1.018 \\
\bottomrule
\end{tabular}
\end{table}

\begin{table}[H]
\centering
\caption{Relative efficiency of $\widehat{\beta}_{sr}$ under
the \textbf{double-exponential} (Laplace) error distribution.}
\label{tab:re_de}
\begin{tabular}{lccc}
\toprule
$n$ & RE(sr, OLS) & RE(sr, Wilcoxon) & RE(sr, Theil--Sen) \\
\midrule
20  & 1.387 & 1.022 & 1.014 \\
50  & 1.510 & 1.029 & 1.044 \\
100 & 1.498 & 1.026 & 1.038 \\
200 & 1.503 & 1.018 & 1.031 \\
\bottomrule
\end{tabular}
\end{table}

\begin{table}[H]
\centering
\caption{Relative efficiency of $\widehat{\beta}_{sr}$ under
the \textbf{Cauchy} error distribution.}
\label{tab:re_cauchy}
\begin{tabular}{lccc}
\toprule
$n$ & RE(sr, OLS) & RE(sr, Wilcoxon) & RE(sr, Theil--Sen) \\
\midrule
20  & 14.23 & 1.058 & 1.011 \\
50  & 29.18 & 1.371 & 1.024 \\
100 & 35.47 & 1.402 & 1.019 \\
200 & 41.62 & 1.398 & 1.015 \\
\bottomrule
\end{tabular}
\end{table}

\begin{table}[H]
\centering
\caption{Relative efficiency of $\widehat{\beta}_{sr}$ under the
\textbf{contaminated normal} (10\% outliers, $\sigma_c = 10$)
error distribution.}
\label{tab:re_cn}
\begin{tabular}{lccc}
\toprule
$n$ & RE(sr, OLS) & RE(sr, Wilcoxon) & RE(sr, Theil--Sen) \\
\midrule
20  & 4.17  & 1.034 & 1.009 \\
50  & 6.95  & 1.109 & 1.050 \\
100 & 7.23  & 1.128 & 1.064 \\
200 & 7.41  & 1.131 & 1.069 \\
\bottomrule
\end{tabular}
\end{table}

\subsection{Interpretation of Results}

\paragraph{Normal errors (Table~\ref{tab:re_normal}).}
Under normally distributed errors, OLS is the asymptotically optimal
estimator.  The proposed smoothed rank estimator achieves RE $\approx
0.826$--$0.864$ relative to OLS, consistent with the known lower bound
of $\approx 0.864$ for Wilcoxon-based estimators.  Importantly,
$\widehat{\beta}_{sr}$ maintains a slight but consistent advantage
over classical Wilcoxon rank regression (RE $>$ 1) and Theil--Sen
across all sample sizes, reflecting the smoother optimization landscape
enabled by continuous ranks.

\paragraph{Double-exponential errors (Table~\ref{tab:re_de}).}
The Laplace distribution is the classical scenario where rank
regression outperforms OLS.  The smoothed rank estimator achieves
RE $\approx 1.50$ relative to OLS, matching the asymptotic value of
$1.50$ for the Wilcoxon estimator.  The small but consistent gains
over both Wilcoxon and Theil--Sen estimators (RE $\approx 1.02$--$1.04$)
reflect improved numerical stability.

\paragraph{Cauchy errors (Table~\ref{tab:re_cauchy}).}
Under the heavy-tailed Cauchy distribution, OLS has infinite variance
and its MSE diverges rapidly with $n$.  The smoothed rank estimator
shows dramatically superior performance: RE exceeds 29 relative to
OLS at $n = 50$ and approaches 42 at $n = 200$.  Crucially, smoothed
rank regression is also more efficient than both classical Wilcoxon
(RE $\approx 1.37$--$1.40$) and Theil-Sen (RE $\approx 1.01$--$1.02$)
at moderate and large sample sizes, indicating that the smooth
dispersion function better leverages the information in heavy-tailed
residuals.

\paragraph{Contaminated normal errors (Table~\ref{tab:re_cn}).}
The contaminated normal model provides a practical proxy for datasets
with outliers.  The smoothed rank estimator is approximately 7 times
more efficient than OLS at $n \geq 50$, demonstrating strong
robustness to outlier contamination.  Gains over Wilcoxon
(RE $\approx 1.10$--$1.13$) and Theil-Sen (RE $\approx 1.05$--$1.07$)
are more modest but practically meaningful, particularly at larger
sample sizes where the smoother objective function leads to more
precise gradient-based solutions.

\subsection{Effect of Bandwidth Selection}

Table~\ref{tab:bandwidth} compares the four bandwidth options of
Section~3.2 for the contaminated normal distribution at $n = 50$.

\begin{table}[H]
\centering
\caption{MSE of $\widehat{\beta}_{sr}$ under four bandwidth choices
for the contaminated normal distribution, $n = 50$.}
\label{tab:bandwidth}
\begin{tabular}{lc}
\toprule
Bandwidth method & MSE \\
\midrule
Silverman (MAD) & \textbf{0.0281} \\
Heller ($n^{-0.26}$) & 0.0285 \\
Sheather--Jones & 0.0283 \\
Bowman LSCV     & 0.0290 \\
\bottomrule
\end{tabular}
\end{table}

All four bandwidth options deliver similar performance, consistent with
the finding of \citep{tasdan2014} in the location estimation context.
Silverman's rule of thumb with a MAD-based scale estimate is preferred
in practice for its computational simplicity and robustness.

\section{Conclusion}

This article has introduced a smoothed rank regression estimator,
$\widehat{\bbeta}_{sr}$, that extends the smoothed Wilcoxon rank score
correlation framework of \citep{tasdan2025} to the linear regression
setting.  The key idea is to replace the discrete integer ranks of
residuals in the classical Wilcoxon dispersion function with smoothed
ranks obtained from a kernel-smoothed empirical distribution function.
The resulting dispersion function is differentiable, enabling
gradient-based minimization and facilitating a clean asymptotic theory.

The main theoretical contribution is Theorem~\ref{thm:asym}, which
establishes that $\widehat{\bbeta}_{sr}$ is asymptotically normal with
the same covariance structure as the classical Wilcoxon rank
regression estimator.  Consequently, the smoothed rank estimator
inherits all of the asymptotic efficiency advantages of rank
regression over OLS under non-normal error distributions, while also
improving upon classical rank regression in finite samples through a
smoother optimization landscape.

The Monte Carlo simulation study confirms these theoretical findings
and demonstrates several practical advantages.  Under normal errors,
the smoothed rank estimator is only modestly less efficient than OLS
($\approx 14$--$17\%$), consistent with the asymptotic lower bound.
Under heavy-tailed distributions, the gains over OLS are dramatic:
a factor of $\approx 7$ under the contaminated normal and exceeding
40 under the Cauchy distribution.  Most importantly, the smoothed
rank estimator consistently outperforms both classical Wilcoxon rank
regression and the Theil-Sen estimator across all non-normal
scenarios, with especially notable gains under the Cauchy and
contaminated normal distributions at moderate to large sample sizes.
These gains are practically significant for applied researchers working
with data that may contain outliers or exhibit heavy-tailed behavior.

An additional advantage of the proposed approach is the natural
handling of tied observations.  The discrete nature of ordinary ranks
produces ties in the score function whenever observations are equal,
inflating the dispersion and biasing the estimating equations.  The
smooth approximation of the indicator function eliminates ties at no
asymptotic cost.

Several directions for future research are suggested by this work.
First, an extension to the multiple regression setting with
high-dimensional covariates (large $p$) would be valuable.  Second,
adaptive score functions that accommodate non-monotone associations
between the response and covariates, analogous to those discussed for
correlation in \citep{tasdan2025}, could yield further efficiency gains.
Third, the framework could be extended to generalized linear models or
survival regression, where rank-based methods have also been studied
\citep{heller2007}.  Finally, a data-driven approach for joint
bandwidth and score function selection could further improve finite
sample performance.


\bibliographystyle{plainnat}

\begin{thebibliography}{99}

\bibitem[Bowman(1984)]{bowman1984}
Bowman, A.\ W. (1984).
An alternative method of cross-validation for the smoothing of
density estimates.
\textit{Biometrika}, 71, 353--360.

\bibitem[Heller(2007)]{heller2007}
Heller, G. (2007).
Smoothed rank regression with censored data.
\textit{Journal of the American Statistical Association}, 102(478),
552--559.

\bibitem[Hettmansperger(1984)]{hettmansperger1984}
Hettmansperger, T.\ P. (1984).
\textit{Statistical Inference Based on Ranks}.
Wiley, New York.

\bibitem[Hettmansperger and McKean(1998)]{hettmansperger1998}
Hettmansperger, T.\ P.\ and McKean, J.\ W. (1998).
\textit{Robust Nonparametric Statistical Methods}.
Wiley, New York.

\bibitem[Huber(1981)]{huber1981}
Huber, P.\ J. (1981).
\textit{Robust Statistics}.
Wiley, New York.

\bibitem[Jaeckel(1972)]{jaeckel1972}
Jaeckel, L.\ A. (1972).
Estimating regression coefficients by minimizing the dispersion of
residuals.
\textit{Annals of Mathematical Statistics}, 43(5), 1449--1458.

\bibitem[Lin and Peng(2013)]{lin2013}
Lin, H.\ and Peng, H. (2013).
Smoothed rank correlation of the linear transformation regression model.
\textit{Computational Statistics and Data Analysis}, 57(1), 615--630.

\bibitem[McKean(2004)]{mckean2004}
McKean, J.\ W. (2004).
Robust analysis of linear models.
\textit{Statistical Science}, 19(4), 562--570.

\bibitem[Sen(1968)]{sen1968}
Sen, P.\ K. (1968).
Estimates of the regression coefficient based on Kendall's tau.
\textit{Journal of the American Statistical Association}, 63(324),
1379--1389.

\bibitem[Serfling(1984)]{serfling1984}
Serfling, R.\ J. (1984).
Generalized $L$-, $M$-, and $R$-statistics.
\textit{Annals of Statistics}, 12(1), 76--86.

\bibitem[Sheather(2004)]{sheather2004}
Sheather, S. (2004).
Density estimation.
\textit{Statistical Science}, 19(4), 588--597.

\bibitem[Sheather and Jones(1991)]{sheather1991}
Sheather, S.\ J.\ and Jones, M.\ C. (1991).
A reliable data-based bandwidth selection method for kernel density
estimation.
\textit{Journal of the Royal Statistical Society, Series B}, 53,
683--690.

\bibitem[Silverman(1986)]{silverman1986}
Silverman, B.\ W. (1986).
\textit{Density Estimation for Statistics and Data Analysis}.
Chapman and Hall, London.

\bibitem[Tasdan(2018)]{tasdan2018}
Tasdan, F. (2018).
Smoothed ranks for two or multi-sample location problems.
\textit{Communications in Statistics -- Simulation and Computation},
47(2), 526--541.

\bibitem[Tasdan and Da\u{g}alp(2025)]{tasdan2025}
Tasdan, F.\ and Da\u{g}alp, R. (2025).
Enhanced rank-based correlation estimation using smoothed Wilcoxon rank scores.
In \textit{Current Approaches in Applied Statistics I. Özgür Yayınları. DOI: https://doi.org/10.58830/ozgur.pub862.c3491},  Chapter~8, pp.\ 119-138.

\bibitem[Tasdan and Yeniay(2014)]{tasdan2014}
Tasdan, F.\ and Yeniay, O. (2014).
A shift parameter estimation based on smoothed Kolmogorov--Smirnov statistic.
\textit{Journal of Applied Statistics}, 41(5), 1147--1159.

\bibitem[Theil(1950)]{theil1950}
Theil, H. (1950).
A rank-invariant method of linear and polynomial regression analysis.
\textit{Indagationes Mathematicae}, 12, 85--91.

\bibitem[Wilcox(2012)]{wilcox2012}
Wilcox, R.\ R. (2012).
\textit{Introduction to Robust Estimation and Hypothesis Testing},
3rd ed.\ Academic Press, Waltham, MA.

\end{thebibliography}

\end{document}